\documentclass[12pt,oneside]{amsart}
\usepackage{amssymb}
\usepackage{amsfonts}
\usepackage{amscd}
\usepackage[matrix, arrow, curve]{xy}
\usepackage{graphicx}
\numberwithin{equation}{section}

\newtheorem{theorem}{Theorem}[section]

\newtheorem{proposition}[theorem]{Proposition}

\newtheorem{lemma}[theorem]{Lemma}
\theoremstyle{definition}

\theoremstyle{remark}
\newtheorem{remark}{Remark}[section]
\newtheorem{example}{Example}[section]

\begin{document}
%%%%%%%%%%%%%%%%%%%%%%%%%%%%%%%%%%%%%%%%%%%%%%%%%
%%%%%%%%%%%%  macrodefinitions
%%%%%%%%%%%%%%%%%%%%%%%%%%%%%%%%%%%%%%%%%%%%%%%%%
% Общие
\newcommand{\define}{\def}
\newcommand{\redefine}{\def}
%%%%%%%%%%%%%%%%%%%%%%%%%%%%%%%%%%%%%%%%%%%%%%%%%
%%%%   temporary
\newcommand{\vac}{|0\rangle}
%%%%%%%%%%%%%%%%%%%%%%%%%%%%%%%%%%%%
%   Referencing Scheme of Martin
%%%%%%%%%%%%%%%%%%%%%%%%%%
\newcommand{\refE}[1]{(\ref{E:#1})}
\newcommand{\refCh}[1]{Chapter~\ref{Ch:#1}}
\newcommand{\refS}[1]{Section~\ref{S:#1}}
\newcommand{\refSS}[1]{Section~\ref{SS:#1}}
\newcommand{\refT}[1]{Theorem~\ref{T:#1}}
\newcommand{\refO}[1]{Observation~\ref{O:#1}}
\newcommand{\refP}[1]{Proposition~\ref{P:#1}}
\newcommand{\refD}[1]{Definition~\ref{D:#1}}
\newcommand{\refC}[1]{Corollary~\ref{C:#1}}
\newcommand{\refL}[1]{Lemma~\ref{L:#1}}
\newcommand{\refEx}[1]{Example~\ref{Ex:#1}}
%%%%%%%%%%%%%%%%%% End Referencing Scheme %%%%%%%%%%%%%%%%
\newcommand{\R}{\ensuremath{\mathbb{R}}}
\newcommand{\C}{\ensuremath{\mathbb{C}}}
\newcommand{\N}{\ensuremath{\mathbb{N}}}
\newcommand{\Q}{\ensuremath{\mathbb{Q}}}
\renewcommand{\P}{\ensuremath{\mathcal{P}}}
\newcommand{\Z}{\ensuremath{\mathbb{Z}}}
%%%%%%%%%%%%%%%%%%%%%%%%%%%%%%%%%%%%%%%%%%
%%%%%%%%%%%%%%%%%%%%%%%%%%%%%%%%%%%%%%%%%%%%%
%%%%%%%%%%%%%%%%%%%%%%%%%%%%%%%%%%%%%%%%%%%%%%%%%%
\newcommand{\g}{\mathfrak{g}}
\newcommand{\h}{\mathfrak{h}}
%%%%%%%%%%%%%%%%%%%%%%%%%%%%%%%%%%%%%%%%%%%%%%%%%%%%
\newcommand{\gl}{\mathfrak{gl}}
\newcommand{\sln}{\mathfrak{sl}}
\newcommand{\so}{\mathfrak{so}}
\newcommand{\spn}{\mathfrak{sp}}
%%%%%%%%%%%%%%%%%%%%%%%%%%%%%%%%%%%%%%%%%%%%%%%%%%%%
\newcommand{\tr}{\mathrm{tr}}
\newcommand{\ord}{\operatorname{ord}}
\newcommand{\res}{\operatorname{res}}
\newcommand{\nord}[1]{:\mkern-5mu{#1}\mkern-5mu:}
\newcommand{\codim}{\operatorname{codim}}
\newcommand{\ad}{\operatorname{ad}}
\newcommand{\Ad}{\operatorname{Ad}}
\newcommand{\supp}{\operatorname{support}}
\define\im{\operatorname{Im}}
\define\re{\operatorname{Re}}
\redefine\deg{\operatornamewithlimits{deg}}
\define\rank{\operatorname{rank}}
\newcommand{\Aut}{\operatorname{Aut}}
\newcommand{\End}{\operatorname{End}}
\newcommand{\Hom}{\rm{Hom}\,}
%%%%%%%%%%%%%%%%%%%
\define\a{\alpha}
\redefine\d{\delta}
\newcommand{\w}{\omega}
\define\ep{\epsilon}
\redefine\b{\beta}
\define\ga{\gamma}
\renewcommand{\l}{\lambda}
\newcommand{\ka}{\kappa}
\renewcommand{\k}{\varkappa}
\define \s{\sigma}
\define\K{\mathcal K}
\define\U{\mathcal U}
\redefine\O{\mathcal O}
\define\A{\mathcal A}
\redefine\L{\mathcal L}
\newcommand{\M}{\mathcal{M}}
\newcommand{\F}{\mathcal{F}}
\redefine\D{\mathcal D^{1}}
\newcommand{\Ga}{\Gamma}
%%%%%%%%%%%%%%%%%%%%
%%%%%%%%%%%%%%%%%%%%%%%%%%
\define \nord #1{:\mkern-5mu{#1}\mkern-5mu:}
\newcommand{\MgN}{\mathcal{M}_{g,N}} %% moduli space
%%%%%%%%%% Olegs definitions %%%%%%%%%%%%%%%%%%%%%%%%%%%%%%%%%%%
\newcommand{\calF}{{\mathcal F}}
\newcommand{\ferm}{{\mathcal F}^{\infty /2}}
%%%%%%%%%%%%%%%%%%%%%%%%%%%%%%%%%%%%%%%%%%
%%%%%%%%%%%%  END of macrodefinitions
%%%%%%%%%%%%%%%%%%%%%%%%%%%%%%%%%%%%%%%%%

\title{Hitchin systems on hyperelliptic curves}
\author{P.I. Borisova, O.K. Sheinman }
\maketitle
\begin{abstract}
We describe a class of spectral curves and find explicit formulas for Darboux coordinates for hyperelliptic Hitchin systems corresponding to classical simple Lie groups. We consider in detail the systems with classical rank 2 gauge groups on genus 2 curves.
\end{abstract}
\tableofcontents
%%%%%%%%%%%%%%%%%%%%%%%%%%%%%%%%%%%%%%%%%%%%%%%%%%%%%%%%%%%%%%%%%%
\section{Introduction}
In the vast literature on Hitchin systems there are not so many works devoted to the fundamental question of separation of variables for them. With this relation, we can only mention \cite{Gaw,Hur,GNR}. The first one treats in detail the system on a genus 2 curve with $SL(2)$ as a gauge group, and in the two others some general approaches to the problem are discussed. The separation of variables becomes effective while making use of the family of spectral curves \cite{DKN,Skl,BT2}. General properties of spectral curves depending on a gauge group have been formulated in the pioneering work by Hitchin \cite{Hitchin}.

A description of spectral curves and separation of variables for hyperelliptic Hitchin systems  and gauge groups  of the series $A_l$, $B_l$, $C_l$ has been given in \cite{Sh_FAN_2019}. These are the classical series for which the spectral curve is nonsingular, and the Hamiltonians can be found out explicitly in terms of separation variables. It is observed in \cite{Sh_FAN_2019} that the $D_l$ series is specific with this regard. For the group $SO(4)$ as a simplest representative of this series the separation of variables is carried out in \cite{PBor}. We reproduce these results in the present paper with the corresponding references. As new results, we prove here a holomorphy of differentials of angle coordinates on the normalization of the spectral curve in the case it is singular (it is just the case of $D_l$) and find the list of basis holomorphic Prym differentials for the systems of type $D_l$. Together with our previous results \cite{Sh_FAN_2019} it gives the full list of basis holomorphic differentials (resp., Prym differentials in the case of spectral curves with involution ) on a generic leaf of the Hitchin foliation for every classical simple group. The rank one and two cases, i.e. the groups $SL(2)$, $SO(4)$, $Sp(4)$ and $SO(5)$ are considered in detail. In \refS{concl} we discuss the relation between our results in the $SL(2)$ case and the corresponding results in \cite{Gaw}. Observe that the description of the class of spectral curves we have given in \cite{Sh_FAN_2019}, relies on the Lax representation of Hitchin systems, proposed and investigated in \cite{Kr_Lax} (see \cite{Kr_Sh_FAN,Shein_UMN2016,Sh_DGr,Sh_TMPh} for arbitrary simple gauge groups). Here, we derive this description from the properties of Higgs fields.

In \refS{Hitchin} we introduce Hitchin's systems following the lines of his pioneering work \cite{Hitchin}, and give a description of spectral curves for the systems on hyperelliptic curves, for all classical root systems.

In \refS{separ} we give an alternative definition of Hitchin systems in terms of Separation of Variables. We find Darboux coordinates and investigate properties of the angle-type variables.

In \refS{Rank2} we consider the rank 1 and 2, and genus 2 examples, i.e. the systems with gauge groups $SL(2)$, $SO(4)$, $Sp(4)$ и $SO(5)$. We give the detailed data of their spectral curves, namely genus, number of sheets, number of branch points, gluing scheme, action of the involution (if any), list of the basis holomorphic differentials (resp., Prym differentials), the form of the Prym mapping. In the cases $SL(2)$, $SO(4)$ we find out the action--angle variables.

In \refS{concl} we discuss the relation between our results in the case $SL(2)$ and the corresponding results in \cite{Gaw}, and also our following observation: for $SL(2)$, the Prymians of spectral curves turn out to be Jacobians (of some other curves).

\bigskip
The authors are grateful to I.A.Taimanov for the discussion of Prymians of singular curves, though it turned out to be beyond the framework of the final text. Among the works especially influenced our present work, we feel obliged to mention \cite{Hitchin,Kr_Lax,Previato}.
%%%%%%%%%%%%%%%%%%%%%%%%%%%%%%%%%%%%%%%%%%%%%%%%%%%%%%%%%%%%%%%%%%
\section{Hitchin systems on hyperelliptic curves}\label{S:Hitchin}
\subsection{Hitchin systems}
Define Hitchin systems following the lines of \cite{Hitchin}.

Assume, $\Sigma$ is a compact genus $g$ Riemann surface with a conformal structure, $G$ is a complex semisimple Lie group, $\g = Lie(G)$, $P_0$ is a smooth principal $G$-bundle on~$\Sigma$.

By \emph{holomorphic structure} on $P_0$ we mean a connection of type $(0, 1)$, i.e. a differential operator on the sheaf of sections of the bundle $P_0$. Locally, the operator is given as $\bar{\partial} + \w$, where $\w \in \Omega^{0, 1}(\Sigma, \g)$ and a gluing function $g$ acts on $\w$ by the gauge transformation:
\[
\w \to g\w g^{-1} - (\bar{\partial}g)g^{-1}.
\]

Suppose $\A$ is a space of semistable \cite{Hitchin} golomorphic structures on $P_0$, ${\mathcal G}$ is a group of smooth global gauge transformations. The quotient ${\mathcal N} = \A/{\mathcal G}$ is called \emph{moduli space of stable holomorphic structures on $P_0$}. Further on, we will consider ${\mathcal N}$ as a configuration space of a Hitchin system. A point in ${\mathcal N}$ is a principal holomorphic $G$-bundle on $\Sigma$. The dimension of ${\mathcal N}$ is $\dim{\mathcal N} = \dim\g\cdot (g-1)$.

By definition \cite{Hitchin}, the \emph{phase spase} of a Hitchin system is $T^{*}({\mathcal N})$. According to Codaira--Spenser theory, $T_P({\mathcal N}) \simeq H^1(\Sigma, \Ad P)$. By Serre duality
\[
  T^*_P({\mathcal N})\simeq H^0(\Sigma, \Ad P\otimes {\mathcal K})
\]
where ${\mathcal K}$ is a canonical class of $\Sigma$, $\Ad P\otimes {\mathcal K}$ is a holomorphic vector bundle with a fiber $\g\otimes\C$. We denote the points of $T^*({\mathcal N})$ by $(P, \Phi)$, where $P \in {\mathcal N}, ~\Phi\in H^0(\Sigma, \Ad P\otimes {\mathcal K})$. Sections of the sheaf $T^*({\mathcal N})$ are called {\it Higgs fields}.

Assume, $\chi_d$ is a homogeneous invariant polynomial on $\g$ of degree $d$. It defines a map $\chi_d(P)\colon H^0(\Sigma, \Ad P\otimes {\mathcal K})\to H^0(\Sigma, {\mathcal K}^d)$ for each $P \in {\mathcal N}$. Let $\Phi$ stay for a Higgs field, then we can define $\chi_d(P, \Phi) = (\chi_d(P))(\Phi(P))$. By that, to each point $(P, \Phi)$ of the phase space we have assigned an element of $H^0(\Sigma, {\mathcal K}^d)$. Suppose $\{\Omega^d_j\}$ is a basis in $H^0(\Sigma, {\mathcal K}^d)$, then $\chi_d(P, \Phi) = \sum H_{d, j}(P,\Phi)\Omega^d_j$, where an $H_{d, j}(P, \Phi)$ is a scalar function on $T^*({\mathcal N})$. For any $j$ and $d$ the function $H_{d, j}(P, \Phi)$ is called a {\it Hitchin's Hamiltonian}.
\begin{theorem}[\cite{Hitchin}]
All Hitchin Hamiltonians Poisson commute on $T^*({\mathcal N})$.
\end{theorem}
%%%%%%%%%%%%%%%%%%%%%%%%%%%%%%%%%%
\subsection{Spectral curves of hyperelliptic Hitchin systems}\label{SS:num_int}
Here we define Hamiltonians in another way. We fix a holomorphic differential $\w$ on hyperelliptic curve $\Sigma$ and divide holomorphic sections of $\Ad P\otimes {\mathcal K}$ by it. By that, we obtain meromorphic sections of the bundle $\Ad P$ with a divisor of poles $-D$, where $D = (\w)$ is the divisor of zeros of differential $D$. In this way, $\chi_d(P)$ will be a map from $H^0(\Sigma, \Ad P, -D)$ to $\O(\Sigma, -dD)$. We will also consider the basis $\Omega^d_j$ as a basis of $\O(\Sigma, -dD)$.

Assume, $\Sigma$ is a hyperelliptic (in particular, non-singular) curve defined by the equation
\begin{equation}
y^2 = P_{2g+1}(x),~\text{where}~P_{2g+1}(x)= x^{2g + 1} + \sum\limits_{i=0}^{2g}a_ix^{i}.
\end{equation}

We choose $\w = \frac{dx}{y}$, hence $D = 2(g - 1)\cdot\infty$. The following lemma enables us to find out a basis of the space $\O(\Sigma, -dD)$.

\begin{lemma}[\cite{Sh_FAN_2019}]\label{L:lemma1}
The functions $\{1, x, \ldots, x^{d_i(g - 1)}\}$ and $\{y, yx, \ldots, yx^{(d_i - 1)(g - 1) - 2}\}$ form a basis of $\O(-dD)$, where $D = 2(g - 1)\cdot\infty$.
\end{lemma}

The \emph{spectral curve} of a Higgs field $\Phi$ is defined by the relation
\[
\det(\l - \Phi/\w) = 0.
\]
Locally, the value of $\Phi(P)$ is a $\g$-valued function on $\Sigma$. Denote its value at a point $(x, y) \in \Sigma$ by $\Phi(P, x, y)$. For a fixed $P$ we consider the equation of the spectral curve as a relation between $\l, x, y$. Letting $P$ to run over the moduli space ${\mathcal N}$, we obtain a family of spectral curves. In different local trivializations, evaluations of $\Phi$ are related by the group $\Ad G$ action, thereby the equation is well-defined.

Assume $d_1, \ldots, d_l, ~l = \rank\g$ be a set of degrees of basis invariant polynomials of the Lie algebra $\g$. For brevity, we will denote $\chi_{d_i}$ by $\chi_i$, and $H_{d_i, j}$ by $H_{i, j}$. The meromorphic functions on $\Sigma$ of the form $p_i = \chi_i(\Phi/\w)$, where $\chi_i~(i = 1, \ldots, l)$ are the basis invariant polynomials, will be referred as \emph{\it basis spectral invariants} (because they are invariant under Hitchin flows). The integer $d_i = \deg\chi_i$ is called \emph{the degree of the basis spectral invariant $p_i$}.

Consider a classical Lie algebra $\g$ in the standard representation. The equation of the spectral curve has the following form:
\begin{equation}\label{E:spec}
R(x, y, \l) = \l^n + \sum\limits_{i=1}^{l} r_i(x, y)\l^{n - d_i}=0,
\end{equation}
where $n$ is the dimension of the standard representation of the Lie algebra $\g$, $r_i~(i = 1, \ldots, l)$ are meromorphic functions on $\Sigma$. Thus, for a Lie algebra $\g$ of type $A_l$ we have $n = l+1, ~d_i = i+1$; for the type $B_l$ we have $n = 2l + 1, ~d_i = 2i$; for the type $C_l$ we have $n = 2l, ~d_i = 2i$, and $r_i = p_i$ are the basis invariants for all $i = 1, \ldots, l$. The case of Lie algebras of type $D_l$ is exceptional with this regard, namely, the coefficient $r_l$ of degree $2l$ in \refE{spec} is a square of the basis invariant of degree $l$ (which is nothing but the Pfaffian of $L$). Looking further forward, we remark that this is the reason why the equations for Hamiltonians in the method of separation of variables are non-linear for Lie algebras of type $D_l$.

A basis invariant $r_i$ of degree $d_i$ can be expanded over the basis from \refL{lemma1} as follows:
\begin{equation}\label{E:eq1} r_{i}(x, y, H) = \sum\limits_{k=0}^{d_i(g - 1)}H^{(0)}_{ik}x^k + \sum\limits_{s = 0}^{(d_i - 1)(g - 1) - 2}H^{(1)}_{is}yx^s,
\end{equation}
where $H^{(0)}_{ik}, ~H^{(1)}_{is}$ are independent Hamiltonians of the Hitchin system.
\begin{example}
A spectral curve for $\g = \sln(2)$. There is only one basis invariant $p_2 = r_2$ of degree $2$
\begin{equation}
r_2(x, y) = \sum\limits_{k = 0}^{2(g-1)}H^{(0)}_{k}x^k + \sum\limits_{s = 0}^{g-3}H^{(1)}_{s}yx^s
\end{equation}
(in particular, for $g=2$ the second sum is absent). The equation of the spectral curve has the form
\begin{equation}
\l^2 + r_2(x, y) = 0.
\end{equation}
\end{example}
\begin{example}
Spectral curve for $\g = \so(4)$ is defined by the following equation:
\begin{equation}
R(x, y, \l, H) = \l^4 +\l^2p + q = 0
\end{equation}
where $p$ and $q$ are the basis spectral invariants of degree 2.
\end{example}
For genus $2$ these examples will be considered  in \refS{Rank2} in detail.
%%%%%%%%%%%%%%%%%%%%%%%%%%%%%%%%%%%%%%%%%%%%%%%%%%%%%%%%%%%%%%%%%%%%%%%%%
\section{Separation of variables for hyperelliptic Hitchin systems}\label{S:separ}
\subsection{Separating variables. Symplectic and Poisson structure}
Denote the number of degrees of freedom of a Hitchin system by $N$. It is known \cite{Hitchin} that $N = \dim\g\cdot(g-1)$ provided $\g$ is a simple Lie algebra.

The spectral curve \refE{spec} is completely defined by a set of values of independent Hamiltonians, hence by $N$ points the spectral curve passes through. Denote these points by $(x_i, y_i, \l_i) ~(i = 1, \ldots, N)$, where $y^2 = P_{2g+1}(x)$ and $x_i, ~y_i, \l_i$ are related by the equations \refE{spec}, \refE{eq1} for every $i$. The variables $x_i$ and $\l_i$ are called \emph{separating variables}.

Define a $2$-form $\s$ as follows:
\begin{equation}\label{E:sympl}
\s = \sum\limits_{i=1}^{N}d\l_i\wedge\frac{dx_i}{y_i}.
\end{equation}
From Theorem 4.3 \cite{Kr_Lax}, one can derive that the form $\s$  gives a symplectic structure identical to that for a Hitchin system.

\begin{remark}
To restore a Hithin system from the above data one should consider the points $(x_i, y_i, \l_i)$ as poles of an eigenfunction of the Lax operator of a Hitchin system, and use the inverse scattering method following the lines of \cite{Kr_Lax}. Certainly, the relation \refE{sympl} is nothing but a reformulation of Theorem 4.3 \cite{Kr_Lax} for the case of a hyperelliptic base curve. In particular, the $\frac{dx_i}{y_i}$ in \refE{sympl} come from the differential $\frac{dx}{y}$, which was fixed in our definition of  Hamiltonians (and of the symplectic structure in \cite{Kr_Lax} as well). A similar expression for the symplectic form occurs also in \cite{DW} for elliptic curves.

Summarizing the above, the phase space of a Hitchin system on a hyperelliptic curve $\Sigma$ is constituted by unordered sets of triples of the following form:
\[
\{(x_i, y_i, \l_i)| i = 1, \ldots, N\},
\]
with the symplectic structure given by \refE{sympl}.

The corresponding Poisson structure is given by the relation
\begin{equation}
\{\l_i, x_j\} = \delta_{ij}y_i.
\end{equation}
\end{remark}
%%%%%%%%%%%%%%%%%%%%%%%%%%%%%%%%%%%%%%%%%%%%%%%%%%%%%%%%%%%%%

\subsection{The form of Hamiltonians in separating variables}
We find the Hamiltonians using the fact that the spectral curve passes through the points $(x_i, y_i, \l_i)$. It can be expressed as a system of equations
\begin{equation}\label{E:syst}
R(x_i, y_i, \l_i, H) = 0, ~ i=1, \ldots, N
\end{equation}
with unknowns
\begin{align*}
H = \{H^{(0)}_{jk}, H^{(1)}_{js}|& j=1, \ldots, l; k=0, \ldots, (2d_j - 1)(g-1);\\
   &  s=0,\ldots, (d_j - 1)(g - 1) - 2\}.
\end{align*}
We will call the equations \refE{syst} {\it the separation relations}. This system is linear for Lie algebras of the types $A_l, ~B_l, ~C_l$ which follows from the relations \refE{spec}, \refE{eq1}, so the Hamiltonians can be explicitly found in terms of the separating variables by Kramer's rule.

\begin{example}
For $\g = \sln(2)$ the system of equations has the following form:
\begin{equation}
\l_i^2 + \sum\limits_{k=0}^{2(g-1)}H^{(0)}_kx^k_i +\sum\limits_{s=0}^{g - 3}H^{(1)}_sx^s_i = 0, ~i = 1, \ldots, 3(g-1).
\end{equation}
In this way, $H^{(0)}_k = D^{(0)}_k/D$, $H^{(1)}_s = D^{(1)}_s/D$, where
\[
D =
\left |
\begin{array}{cccccc}
    1 & \ldots & x^{2(g - 1)_1} & y_1 & \ldots & y_1x^{g - 3}_1 \\
   \vdots & \vdots & \vdots & \vdots & \vdots & \vdots\\
    1 & \ldots & x^{2(g - 1)}_{3(g - 1)} & y_{3(g - 1)} & \ldots & y_{3(g - 1)}x^{g - 3}_{3(g - 1)}
\end{array}
\right |,
\]
$D^{(0)}_k$ is obtained  by replacing $x^k_i$ with $(-\l^2_i)$ in the $k$-th column of the determinant $D$, and $D^{(1)}_s$ is obtained  by replacing $y_ix^k_i$ with $(-\l^2_i)$ in the $(2g - 1 + k)$-th column of $D$.
\end{example}

In the case of $D_l$ series the system of separation relations is quadratic. In \cite{PBor}, it was shown that the system is unsolvable in radicals  except for $l = 2$ (which corresponds to the Lie algebra $\so(4)$). The case $l = 2$ descends to resolving an algebraic equation of degree 4 with one unknown (see  \refS{Rank2}). Note that the isomorphism $\so(4) \cong \so(2)\times \so(2)$ does not simplify the problem since it is an outer isomorphism not preserving the spectral curve.
%%%%%%%%%%%%%%%%%%%%%%%%%%%%%%%%%
\subsection{Darboux coordinates}
We will use also the following through enumeration of the Hamiltonians: $H=(\ldots, H_j, \ldots)$. Our next goal is to find the coordinates $\phi_j$ conjugate to Hamiltonians. These are coordinates in a covering of a generic leaf of the Lagrangian foliation $H = const$ (which is called {\it Hitchin foliation} in the case of Hitchin systems). They can be found out in a standard way by means of the technique of \emph{generating functions} \cite{Arnold}. Finally, we obtain the following result:
\begin{equation}\label{E:angle}
\phi_j = \sum\limits_{i=0}^{N}\int\limits^{(x_i, y_i,\l_i)}\w_j, ~\text{where } \w = \frac{\partial R/\partial H_j}{\partial R/\partial \l}\frac{dx}{y}
\end{equation}
(the calculation comes back to \cite{Hur}, we refer to \cite{Sh_FAN_2019} for details; see also \cite[eq. (4.61)]{Kr_Lax}).

The coordinates $(H_j, \phi_j)$ possess the Darboux property that immediately follows from the method of generating functions. In case the equation of a spectral curve is linear in $H$-coordinates, it follows also from the results of the articles \cite{BT2,DT}. For the root system $D_l$ these results do not work.

Next, we consider properties of Abelian differentials $\w$.
\begin{proposition}\label{P:h_diff}
Assume that a spectral curve has no worse than simple singularities, and projections of its ramification divisors and of its divisor of singularities do not intersect with the ramification divisor of the base curve over $\C P^1$. Then, away from the infinity, the differentials $\w_j$ are holomorphic at smooth points of the spectral curve. If the spectral curve is singular, their pull-back onto the normalization results in holomorphic differentials (away from the infinity).
\end{proposition}
\begin{proof}
If $R'_{\l} \neq 0$, then $\w_j$ is holomorphic because $dx/y$ is holomorphic and $\partial R/\partial H_j$ is a polynomial in  $x, ~y, ~\l$. If $R'_{\l} = 0$, i.e. the point $(x, ~y, ~\l)$ is a finite branch point, and, moreover, this point is non-singular, then $R'_x \neq 0$. From the equation $R'_xdx + R'_{\l}d\l = 0$ we obtain $\frac{dx}{R'_{\l}} = -\frac{d\l}{R'_x}$, and the right hand side of the relation is holomorphic. Therefore the left hand side is also holomorphic. In addition, by assumption, this point is not a branch point for the base curve, hence $y \neq 0$. Hence at a finite branch point $\w_j$ is holomorphic for every $j$.

In the case of a singular point, it is a simple singularity by assumption. Locally, in the neighborhood of the singularity,  the equation of the curve has the form $R = g_1\cdot g_2$. Normalization of the curve splits into two smooth branches $\eta_1, ~\eta_2$ (see Fig.~\ref{nodal_pic}).
%%%%%%%%%%%%%%%%%%%%%%%%%%%%%%%%%
% рис. - сингулярная особенность
\begin{figure}[h]
\begin{picture}(150,60)
\unitlength=0.6pt
\thicklines
\qbezier(50,90)(130,80)(180,30)% \eta_1
\put(25,88){$\eta_1$}
\put(185,28){$\eta_1$}
\qbezier(80,100)(130,80)(150,30)% \eta_2
\put(93,100){$\eta_2$}
\put(125,30){$\eta_2$}
\end{picture}
\caption{}\label{nodal_pic}
\end{figure}
%%%%%%%%%%%%%%%%%%%%%%%%%%%%%%%%%%%

\noindent
The first branch is given by the relation $g_1 = 0$, and $g_2 \neq 0$, and for the second branch $g_1 \neq 0$, $g_2 = 0$.

On the first branch, we have
\[
\w_j = \frac{(g_1)'_{H_j}g_2 + g_1(g_2)'_{H_j}}{(g_1)'_{\l}g_2 + g_1(g_2)'_{\l}}\frac{dx}{y}|_{\eta_1} = \frac{(g_1)'_{H_j}}{(g_1)'_{\l}}\cdot\frac{dx}{y}.
\]
The latter differential is holomorphic (locally, in a neighborhood of the point) for the same reason as in the first part of the theorem. Similarly, pull-back of $\w_j$ onto the second branch of the normalizing curve is also holomorphic. In this way, we obtain that differentials $\w_j$ are globally holomorphic on the normalized curve (outside infinity).
\end{proof}

\begin{proposition}\label{P:prop1}
$1^\circ$. In the case of series $A_l, ~B_l, ~C_l$ the full list of the differentials $\w_j$ is as follows:
\begin{equation}\label{E:diff1}
\frac{x^k\l^{n - d_i}dx}{R'_{\l}(x, y, \l)y}(0\leq k \leq d_i(g-1))\text{~and~}\frac{x^s\l^{n - d_i}dx}{R'_{\l}(x, y, \l)}(0\leq s \leq (d_i - 1)(g-1) - 2),
\end{equation}
where $i = 1, \ldots, l$. For the series $D_l$ the differentials are the same for $i < l$ and for $i=l$ they are as follows:
\begin{equation}\label{E:diff2}
\frac{x^kqdx}{R'_{\l}(x, y, \l)y}(0\leq k \leq l(g-1))\text{~and~}\frac{x^sqdx}{R'_{\l}(x, y, \l)}(0\leq s \leq (l - 1)(g-1) - 2)
\end{equation}
where $q = q(x, y)$ is a Pfaffian of the Lax operator.
\vskip4pt
$2^\circ$. The differentials $\w_j$ are holomorphic at infinity. In the case of the systems $A_l$ ($l \geq 2$), together with the differentials $\frac{x^pdx}{y}$ ($p=0,1,\ldots,g-1$) lifted from the base, they form a basis of holomorphic differentials on the spectral curve. For the systems $A_1, ~B_l,~C_l$ they form a basis of holomorphic Prym differentials on the spectral curve with respect to the involution $\l \to -\l$. In the $~D_l$ case they form a basis of holomorphic Prym differentials on the normalization of the spectral curve.
\end{proposition}
\begin{proof}
In the cases $A_l$, $B_l$, $C_l$ it was shown in \cite{Sh_FAN_2019}, in particular the formulae \refE{diff1} are obtained there. The formulae \refE{diff2} follow from \refE{angle} in a similar way.

For the system $D_l$ one should check the holomorphy of differentials \refE{diff2}. As in the proof of a Proposition 3.4 \cite{Sh_FAN_2019}, $q\l^{-1}$ is holomorphic at infinity, hence $q \sim z^{-2l(g-1)}$. Plugging the last, together with the asymptotic expressions  for the other components \cite{Sh_FAN_2019}, we obtain the statement.

In \cite{Hitchin} it was shown that the dimension of the Jacobian of a spectral curve is equal to $n^2(g-1)+1$ ($n=l+1$) in the case of $A_l$, and the dimension of its Prym variety (or Prym variety of its normalization) in other cases is equal to $\dim\g\cdot(g-1)$. It matches the number of differentials $\w_j$ (joined with the differentials lifted from the base in the case $A_l,~l\ge 2$) as it was essentially shown in Proposition 2.3($3^\circ$)\cite{Sh_FAN_2019}.
\end{proof}
%%%%%%%%%%%%%%%%%%%%%%%%%%%%%%%%%%%%%%%%%%%%%%
\subsection{Action-angle coordinates}\label{SS:du}
Among Darboux coordinates, the algebraic-geometric action-angle coordinates play a special role. By \emph{algebraic-geometric angle coordinates} we mean coordinates of the form \refE{angle} subject to the condition that the corresponding differentials are holomorphic and normalized. The angle coordinates are coordinates in a generic leaf of the Hitchin foliation. By \emph{action coordinates} we mean the conjugate coordinates. For $\g=\gl(n)$ the \refE{angle} is nothing but the Abel transform, and for $\g=\sln(2),\so(2l),\so(2l+1),\spn(2l)$ it is the Prym transform of the points in the phase space.

Consider a problem of finding the algebraic-geometric angle coordinates in the case of a smooth spectral curve, i.e for Lie algebras $\sln(n), ~\spn(2l)$ and $\so(2l+1)$ (in the last case we take the non-trivial irreducible component of the spectral curve). According to \refP{prop1}, the differentials listed in assertion $1^\circ$ of the proposition form a basis of holomorphic differentials (holomorphic Prym differentials). As it was shown in \cite{Sh_FAN_2019}, the normalization of this basis is given by a transformation of the form
\[
\w \to A\w, ~H\to HA^{-1},
\]
where $\w = (\w_1, \ldots, \w_N)^T, ~H=(H_1, \ldots, H_N)$ ($^T$ stands for ``transposed"). At the same time, the Darboux property holds true. A standard expression for the matrix $A^{-1}$ is as follows
\begin{equation}\label{E:invmatr}
A^{-1} = \left(2\int\limits_{l_i} \w_k\right)_{i, k = 1, \ldots, N}.
\end{equation}
Here, $l_i~ (i = 1, \ldots, N)$ are the cuts connecting pairs of branch points. According to \cite{Sh_FAN_2019}, the spectral curve can be represented as a result of gluing the $n$ copies of the base curve (\emph{the sheets}) along $\nu/2$ cuts (where $\nu$ is a number of branch points of the spectral curve over the base curve) connecting some pairs of branch points (for an arbitrary pairing). However, independent cycles are given by just $\nu/2 - n + 1$ cuts (denote them by $l_1, \ldots, l_{\nu/2 - n + 1}$). Every sheet is obtained by gluing two copies of $\C P^1$ along $g + 1$ cuts (in this way, we obtain $g$ independent cycles on each sheet). A total number of the cycles is $\nu/2 - n + 1 + ng = \hat{g}$ (by the Riemann-Hurwitz formula). We can arbitrarily choose $N$ independent ones of them. The schemes of gluing the spectral curves from the sheets are given in the next section.

As follows from above, to compute the matrix $A$ (that would resolve the problem of normalization) we need to know branch points of the spectral curve over a base curve, and branch points of the base curve over $\C P^1$. The latter are known but the first bumps into resolving the system of equations $R(\l, x, y) = 0, ~R'_{\l}(\l, x, y) = 0$. In general, it is a system of algebraic equations which turns out to be too complicated to express the matrix $A$ in terms of coordinates $(\l_i, x_i, y_i)$ effectively. However for $\g = \sln(2)$ and $\g = \so(4)$ (and genus $2$) it is possible to solve this problem explicitly. These examples and other examples of the rank 2 Hitchin systems for classical Lie algebras will be considered in the next section.
%%%%%%%%%%%%%%%%%%%%%%%%%%%%%%%%%%%%%%%%%%

\section{Systems of rank 1 and 2 on hyperelliptic curves of genus 2}\label{S:Rank2}
By the rank of a system we mean the rank of the corresponding Lie algebra. Thus, in this section we will consider Hitchin systems for the Lie algebras $\sln(2), ~\so(4), ~\spn(4)$ and $\so(5)$. Ranks of the corresponding bundles are $2, ~4, ~4, ~5$, resp.

\subsection{Hitchin system for the Lie algebra $\sln(2)$ on a hyperelliptic curve of genus~2}\label{SS:sln2}
A Hitchin system for the Lie algebra $\sln(2)$ on a hyperelliptic curve of genus $g$ has $N = 3(g - 1)$ Hamiltonians. For $g = 2$ we obtain the system with $3$ Hamiltonians. Invariants of the Lie algebra $\sln(2)$ descend to only a second order invariant, hence a spectral curve for this case has the form
\begin{equation}
R(\l, x, y) = \l^2 + (H_0 + xH_1 + x^2H_2) = 0.
\end{equation}
The curve is smooth because equations of singular points  ($R'_{\l} = 0$ and $R'_x = 0$) descend to $r_2(x) = 0$ and $(r_2(x))' = 0$, and are incompatible unless $r_2(x)$ has a multiple root.

The spectral curve is a two-sheeted branch covering of the base curve. The branch points can be obtained as solutions of the equation $R'_{\l} = 0$, which descends to $r_2(x) = 0$. Due to the symmetry in $y$, the number of solutions of this equation, that is the number of branch points, is equal to 4. A genus $\hat{g}$ of the spectral curve can be obtained by means the Riemann-Hurwitz formula:
\[
2\hat{g} - 2 = 2(2g - 2) + 4.
\]
For $g = 2$ we obtain $\hat{g} = 5$.

In turn, the base curve is a 2-sheeted branch covering of a Riemann sphere. Thus we can consider the spectral curve as a 4-sheeted covering of the sphere. Moreover, there are 8 branch points at every sheet (6 ones coming from gluing the base curve, and two ones coming from gluing the two copies of the base curve). The Riemann surface obtained in this way is represented at the Figure \ref{sl2_pic}, where circles mean Riemann spheres and every line is a gluing along the cut connecting a pair of branch points (one can interpret the lines as tubes connecting the spheres):
%%%%%%%%%%%%%%%%%%%%%%%%%%%%%%%%%%%%%
% \input{sl2}
\begin{figure}[h]
\begin{picture}(150,50)
\unitlength=0.6pt
\thicklines
\thinlines
\put(10,20){\circle{20}}
\put(10,80){\circle{20}}
\put(10,30){\line(0,1){40}}
\put(3,28){\line(0,1){44}}
\put(16,28){\line(0,1){44}}
\put(80,20){\circle{20}}
\put(80,80){\circle{20}}
\put(80,30){\line(0,1){40}}
\put(73,28){\line(0,1){44}}
\put(86,28){\line(0,1){44}}
%\put(0,55){$P_1$}
%
\put(20,20){\line(1,0){50}}
\put(20,80){\line(1,0){50}}
\end{picture}
\caption{Spectral curve in the case $\sln(2)$, $g=2$}\label{sl2_pic}
\end{figure}
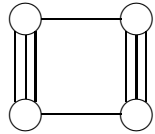

%%%%%%%%%%%%%%%%%%%%%%%%%%%%%%%%%%%%%

The spectral curve is invariant under the involution $\l \to -\l$ (at the Figure, it is a rotation for the angle $\pi$ with respect to the vertical axis). Fixed points of the involution correspond to the branch points, defined by the condition $\l = 0$. In other words, they correspond to the branch points of the spectral curve over the base curve (i.e. to the horizontal lines at the Figure).

The basis of Prym differentials, according to \refP{prop1}, is given by $\w_i = (x^idx)/(2y\l)$, $i = 0, 1, 2$. To calculate a matrix $A^{-1}$ with the help of the relation \refE{invmatr}, we can chose cuts in two ways. One way is to choose three cuts corresponding to vertical lines. Another way is to choose two cuts corresponding to vertical lines and a cut corresponding to a horizontal one. The lower limit of integration in \refE{angle}, by definition of the Prym map, is chosen at one of the fixed points.
%%%%%%%%%%%%%%%%%%%%%%%%%%%%%%%%%%
\subsection{Hitchin system for the Lie algebra $\so(4)$ on a hyperelliptic curve of genus~2}\label{SS:so4}
In this case, the spectral curve is singular. According to \cite{Hitchin}, the Hitchin foliation is a foliation with Prym varieties of normalized spectral curves as leaves.

The Hitchin system in question has $N = 6$ independent Hamiltonians. Lie algebra $\so(4)$ has two basis invariants $p$ and $q$. They both are the second degree invariants (one of them is the Pfaffian, we denote it by $q$). The spectral curve has the form
\begin{equation}\label{E:so4}
R(\l, x, y, H) = \l^4+ \l^2p+q^2 = 0,
\end{equation}
where $p = H_0+xH_1+x^2H_2$, $q = H_3+xH_4+x^2H_5$. Singular points are the solutions of the following system, supplemented by the equation \refE{so4}:
\begin{equation}\label{E:so4syst}
\left \{
\begin{array}{l}
     R'_{\l}(\l, x, y, H) = \l(4\l^2+2p) = 0,\\
     R'_{x}(\l, x, y, H) = \l^2\cdot p'_{x}+2q\cdot q'_{x} = 0.
     \end{array}
\right.
\end{equation}
In a generic position, the spectral curve has $4$ singular points given by the equations $\l=0$, $q = 0$.

In case of $\l^2 = -p/2$ and generic $H$ the second equation of the system \refE{so4syst} does not hold and, thus, we obtain branch points. By \refE{so4} they satisfy the following system of equations:
\begin{equation}\label{E:so4rp}
\left \{
\begin{array}{c}
     \l^2 = -\frac{p}{2},\\
     q = \pm\frac{p}{2}.
     \end{array}
\right.
\end{equation}
Due to symmetries in $y$ and $\l$ we obtain $16$ branch points. For a singular curve, the Riemann-Hutwitz formula gives a genus of its normalization (we go on denoting it by~$\hat{g}$), for this reason we have
\[
2\hat{g} - 2 = 4(2g - 2) + 16
\]
for our spectral curve, which gives $\hat{g} = 13$ for $g = 2$.

The normalized curve is a $4$-sheeted branch covering of the base curve. Having been considered as a branch covering of a sphere, it has 8-sheets, with 10 branch points on every sheet (6 points for gluing the curve of genus 2, and 4 points for gluing the two copies of thus obtained sheets ). The normalized curve is drown at the Figure \ref{so4_pic}.
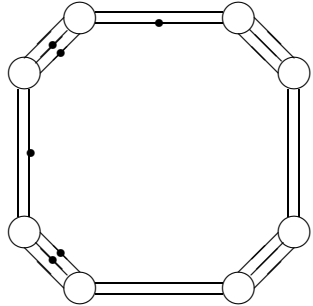
\begin{figure}[h]
\begin{picture}(200,130)
\unitlength=0.6pt
\thicklines
\thinlines
%нижняя половина
\put(55,20){\circle{20}}
\put(20,55){\circle{20}}
\put(47,28){\line(-1,1){20}}
\put(55,30){\line(-1,1){25}}
\put(45,20){\line(-1,1){25}}
\put(155,20){\circle{20}}
\put(190,55){\circle{20}}
\put(163,28){\line(1,1){20}}
\put(155,30){\line(1,1){25}}
\put(165,20){\line(1,1){25}}
\put(65,23){\line(1,0){80}}
\put(65,16){\line(1,0){80}}
%
% верхняя половина
\put(55,190){\circle{20}}
\put(20,155){\circle{20}}
\put(47,182){\line(-1,-1){20}}
\put(55,180){\line(-1,-1){25}}
\put(45,190){\line(-1,-1){25}}
\put(155,190){\circle{20}}
\put(190,155){\circle{20}}
\put(162,183){\line(1,-1){20}}
\put(155,180){\line(1,-1){25}}
\put(165,190){\line(1,-1){25}}
\put(65,187){\line(1,0){80}}
\put(65,194){\line(1,0){80}}
\put(23,65){\line(0,1){80}}
\put(16,65){\line(0,1){80}}
\put(193,65){\line(0,1){80}}
\put(186,65){\line(0,1){80}}
\put(105,187){\circle*{5}}
\put(24,105){\circle*{5}}
\put(43,168){\circle*{5}}
\put(38,173){\circle*{5}}
\put(43,42){\circle*{5}}
\put(38,37){\circle*{5}}
\end{picture}
\caption{The normalized spectral curve in the case $\so(4)$, $g=2$}\label{so4_pic}
\end{figure}
%%%%%%%%%%%%%%%%%%%%%%
The involution $\s\colon\l\to-\l$ looks at the Figure as a rotation around the center of the picture for the angle $\pi$. It has no fixed points. Inspite the singular points of the spectral curve are fixed, their preimages under the normalization map are permuted by the involution. At the Figure, the preimages are located in the middles of the tubes corresponding to the horizontal lines, the two ones on each. One can think of the normalization map as of the identification of points on the opposite horizontal lines.

As a consequence of \refE{angle} and \refP{prop1}, the basis of Prym differentials on the normalized curve is given by
\[
\w^{(0)}_i = \frac{x^{i-1}q(x)dx}{y\l(4\l^2+2p(x))},~ i = 1, 2, 3,
\]
\[
\w^{(0)}_i = \frac{\l x^{i-4}q(x)dx}{y\l(4\l^2+2p(x))},~ i = 4, 5, 6
\]
(see \refE{diff1}, \refE{diff2}). The cuts needed for normalizing the basis are uniquely defined and correspond to the lines marked with dots at the Figure \ref{so4_pic}.

Let $\w$ stay for the set of normalized Prym differentials, $Q_1,Q_2$ stay for a pair of points permuted by the involution. We can write the Prym map in the following way:
\[
\eta(P) = \frac{1}{2}\int\limits_{\gamma}\w,
\]
where $\ga = 2\gamma_1 +\rho$, $\gamma_1$ is an arbitrary path from $Q_1$ to $P$, $\rho$ is a fixed path from $Q_2$ to $Q_1$, hence $\rho+\ga_1$ is a path from $Q_2$ to $P$, see Figure \ref{prym_pic}). The ambiguity in choosing the path  descends to choosing $\gamma_1$ and leads to change of the integral by a double lattice of periods. Thus, $\eta$ is a well-defined map to the Jacobian of the spectral curve. It remains to show that this map is skew-symmetric. For any path $\pi$ and a Prym differential $\w$ we have
$
\int\limits_{\pi}\w = -\int\limits_{\s(\pi)}\w.
$
Indeed, both sides of this relation are equal to $\int\limits_{\s(\pi)}\w^{\s}$: the left hand side by invariance of the integral under a change of variables, and the right hand side because the differential is skew-symmetric: $\w^{\s}=-\w$. We apply this to the proof of the relation $\eta(\s(P)) = -\eta(P)$. As soon as the path $\gamma_1$ from $Q_1$ to $P$ is chosen, take the path from $Q_1$ to $\s(P)$ in the form $\gamma' = \s(\rho) + \s(\gamma_1)$ (Figure~\ref{prym_pic}).
%%%%%%%%%%%%%%%%%%%%%%%%%%%%%%
%\input{Prym_pic}
\begin{figure}[h]
\begin{picture}(150,80)
\unitlength=0.6pt
\thicklines

\put(60,30){\circle*{6}}
\put(60,110){\circle*{6}}
\put(180,30){\circle*{6}}
\put(180,110){\circle*{6}}
\put(60,30){\line(0,1){80}}
\put(180,30){\line(0,1){80}}
\qbezier(60,30)(120,0)(180,30)
\qbezier(60,30)(120,60)(180,30)
%\put(60,30){\line(1,0){10}}
\put(110,40){\vector(1,0){20}}
\put(130,10){\vector(-1,0){20}}
\put(115,22){$\rho$}
\put(110,52){$\s(\rho)$}
\put(55,120){$P$}
\put(155,120){$\s(P)$}
\put(55,10){$Q_1$}
\put(175,10){$Q_2$}
\put(40,70){$\ga_1$}
\put(185,70){$\s(\ga_1)$}
\end{picture}
\caption{}\label{prym_pic}
\end{figure}
%%%%%%%%%%%%%%%%%%%%%%%%%%%%%
\noindent Then $2\eta(\s(P))=\int_{2\ga'+\rho}\w$. We have: $2\ga'+\rho=2\s(\ga_1)+\s(\rho)+(\rho+\s(\rho))=\s(\ga)+(\rho+\s(\rho))$. By the above remark $\int_{\rho+\s(\rho)}\w=0$, hence $2\eta(\s(P))=\int_{\s(\ga)}\w=-\int_\ga\w=-2\eta(P)$.

It is substantial that the system \refE{so4},\refE{so4rp} is biquadratic, for this reason the branch points can be found out explicitly. The Hamiltonians of the system also can be found out explicitly in this case due to the following statement.
\begin{proposition}[\cite{PBor}]
For $\g = \so(4)$, $g=2$ the system of separation equations descends to one algebraic equation of degree 4 in one variable, consequently it is solvable in radicals.
\end{proposition}
Thus, \emph{for $\g=\so(4)$, $g=2$ the problem of finding out the algebraic-geometric action-angle coordinates has been solved as explicitly as possible}, in the sense that finding the action coordinates has been reduced to the solution of a degree 4 equation, and the angle coordinates are expressed in terms of Abelian integrals in the known integration limits.

%%%%%%%%%%%%%%%%%%%%%%%%%%%%%%%%%
\subsection{Hitchin system for the Lie algebra $\spn(4)$ on a hyperelliptic curve of genus~2}
The number of degrees of freedom of this system is $10$. Lie algebra $\spn(4)$ has two basis invariants of degrees $2$ and $4$. Thus the equation of the spectral curve is
\begin{equation}\label{E:sp4}
R(\l, x, y, H) = \l^4+ \l^2r_2(x)+r_4(x) = 0,
\end{equation}
where, due to \refE{eq1},
\begin{align*}
       &r_2(x) = H_0 + xH_1 + x^2H_2,\\
       &r_4(x) = H_3 + \ldots + x^4H_7 + yH_8 + yxH_9.
\end{align*}
Similar to the $\sln(2)$ case, it is a smooth curve, because in a general position the system of equations for singular points has no solutions. Indeed, the system is as follows:
\begin{equation}\label{E:origsp4}
\left \{
\begin{array}{rclcc}
    R(\l, x, y, H) &=& \l^4 + \l^2r_2(x) + r_4(x) &=& 0, \\
    R'_{\l}(\l, x, y, H) &=& \l(4\l^2 + 2r_2(x))&=&0,\\
    R'_{x}(\l, x, y, H) &=& \l^2\cdot(r_2)'_x + (r_4)'_x &=& 0.
     \end{array}
\right.
\end{equation}
From the second equation either $\l = 0$ or $\l^2 = -\frac{r_2}{2}$. If $\l = 0$ the system takes the following form:
\[
\left \{
\begin{array}{rcl}
     \l &=& 0,\\
     r_4(x) &=& 0, \\
     (r_4)'_x &=& 0.
     \end{array}
\right.
\]
It can have solutions only in the case the $r_4(x)$ has multiple roots but this is not true in a generic position. In the case $\l^2 = -\frac{r_2}{2}$, the system \refE{origsp4} descends to the equation $r_2(x) = const$. This is also not the case in a generic position.

Branch points of the spectral curve are given by the first two equations of the system~\refE{origsp4}. For $\l = 0$ one has $r_4(x) = 0$ that is equivalent to the  equation of degree 8, hence gives 8 branch points. For $\l^2 = -\frac{r_2}{2}$ we obtain $\frac{r_2(x)^2}{2} = r_4(x)$ that also leads to an 8 degree equation, which would also result in 8 branch points, but by the symmetry in $\l$ the number of those branch points redoubles, i.e. it is equal to $16$. The total number of branch points is equal to $24$.

By Riemann-Hurwitz formula we obtain the following for the genus $\hat{g}$ of the spectral curve:
\[
2\hat{g} - 2 = 4(2g - 2) + 24,
\]
that gives $\hat{g} = 17$ for $g = 2$.
%%%%%%%%%%%%%%%%%%%%%%%%%%%%%%%%
%\input{SP4}
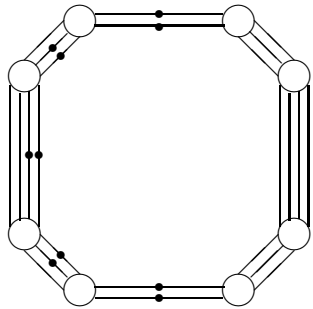
\begin{figure}[h]
\begin{picture}(200,130)
\unitlength=0.6pt
\thicklines
\thinlines
% нижняя половина
% наклонные: юго-запад
\put(55,20){\circle{20}}
\put(20,55){\circle{20}}
\put(47,28){\line(-1,1){20}}
\put(55,30){\line(-1,1){25}}
\put(45,20){\line(-1,1){25}}
% наклонные: юго-восток
\put(155,20){\circle{20}}
\put(190,55){\circle{20}}
\put(163,28){\line(1,1){20}}
\put(155,30){\line(1,1){25}}
\put(165,20){\line(1,1){25}}
%%% нижние горизонтальные
\put(64,22){\line(1,0){82}}
\put(65,15){\line(1,0){80}}
%\put(64,12){\line(1,0){82}}
\put(105,15){\circle*{5}}
\put(105,22){\circle*{5}}
%
% верхняя половина
% наклонные: северо-запад
\put(55,190){\circle{20}}
\put(20,155){\circle{20}}
\put(47,182){\line(-1,-1){20}}
\put(55,180){\line(-1,-1){25}}
\put(45,190){\line(-1,-1){25}}
% наклонные: северо-восток
\put(155,190){\circle{20}}
\put(190,155){\circle{20}}
\put(162,183){\line(1,-1){20}}
\put(155,180){\line(1,-1){25}}
\put(165,190){\line(1,-1){25}}
%%% верхние горизонтальные
\put(64,187){\line(1,0){82}}
\put(65,194){\line(1,0){80}}
%\put(64,197){\line(1,0){82}}
\put(105,194){\circle*{5}}
\put(105,186){\circle*{5}}
% левые вертикальные
\put(11,60){\line(0,1){89}}
\put(17,64){\line(0,1){80}}
\put(23,65){\line(0,1){80}}
\put(29,60){\line(0,1){89}}
\put(29,105){\circle*{5}}
\put(23,105){\circle*{5}}
% правые вертикальные
\put(181,60){\line(0,1){89}}
\put(187,64){\line(0,1){80}}
\put(193,65){\line(0,1){80}}
\put(199,60){\line(0,1){89}}
%\put(183,64){\line(0,1){82}}
%\put(190,65){\line(0,1){80}}
%\put(197,64){\line(0,1){82}}
%

\put(43,168){\circle*{5}}
\put(38,173){\circle*{5}}
\put(43,42){\circle*{5}}
\put(38,37){\circle*{5}}

\end{picture}
\caption{Spectral curve for $\spn(4)$, $g=2$}\label{Sp4_pic}
\end{figure}
%%%%%%%%%%%%%%%%%%%%%%%%%%%%%%%%%%

\noindent
The basis of holomorphic Prym differentials is given by
\[
\w^{(0)}_i = \frac{\l^2x^{i}dx}{y\l(4\l^2+2r_2)},~ i = 0, 1, 2,
\]
\[
\w^{(1)}_i = \frac{\l x^{i-3}dx}{y\l(4\l^2+2r_2)},~ i = 3, \ldots, 7
\]
\[
\w^{(2)}_i = \frac{\l x^{i-8}dx}{\l(4\l^2+2r_2)},~ i = 8, 9.
\]
The involution $\s$ operates as a reflection in the vertical axis at Figure~\ref{Sp4_pic}. To normalize the basis of differentials, one should choose $N = 10$ cuts which are not permuted by the involution, for instance, the cuts corresponding to the lines marked by means the points at the picture (Figure~\ref{Sp4_pic}).

The involution has $8$ fixed points located on the horizontal lines at the picture.

To construct the Prym map one should choose an arbitrary fixed point of the involution. Denote it by $Q$. Then the Prym map is as follows:
\[
\eta(P) = \int_{Q}^{P}\w
\]
where $\w$ is a set of normalized differentials, as above.
%%%%%%%%%%%%%%%%%%%%%%%%%%%%%%%%%%%%%%%%%%%%%%%%%%%%%%%%%%%%%%
\subsection{Hitchin system for the Lie algebra $\so(5)$ on a hyperelliptic curve of genus~2}
This system is locally isomorphic to the previous one \cite{Hitchin}. It also has $10$ degrees of freedom and two basis invariants of degrees $2$ and $4$. The spectral curve has the form
\[
R(\l, x, y, H) = \l^5+ \l^3r_2(x)+\l r_4(x) = 0.
\]
It splits to two irreducible components: $\l = 0$ and $\l^4+ \l^2r_2(x)+r_4(x) = 0$. The first component is trivial, the second one has the same form as the spectral curve for $\spn(4)$. Thus a local description of this system completely descends to the $\spn(4)$ case.
%%%%%%%%%%%%%%%%%%%%%%%%%%%%%%%%%%%%%%%%%%%%%%%
\section{Discussion}\label{S:concl}
In this section, we would like to compare our results with the results of  \cite{Gaw} pursuing similar goals and, in some sense, holding a record since 1998 till 2018. In \cite{Gaw}, a specific Lax representation giving the Hitchin system in the simplest case of the gauge group $SL(2)$ and a base curve of genus $2$ is used. The spectral curve in \cite{Gaw} is different from the standard spectral curve of the system, in particular, from our spectral curve. The standard curve has genus $5$ and is a covering of the spectral curve in \cite{Gaw} which is hyperelliptic of genus $3$. It has been proved in \cite{Gaw} that in the case of the structure group $SL(2)$ and a base curve of genus $2$
\begin{enumerate}
\item the equation of the spectral curve in \cite{Gaw} is a term-by-term product of the equation of the standard  spectral curve by the equation of the base curve;
\item Prym differentials on the standard curve are liftings of all holomorphic differentials on the spectral curve in \cite{Gaw}.
\end{enumerate}
Further on we will formulate the generalization of this result for an arbitrary genus (a gauge group is still $SL(2)$).
\begin{proposition}
\begin{enumerate}
\item Multiply the equations of the spectral curve and of the base curve term by term. Then the normalization of the obtained curve has genus $3(g - 1)$. It is equal to the dimension of the corresponding Prym variety.
\item Prym differentials on the spectral curve can be obtained as a lifting of all holomorphic differentials on this curve.
\end{enumerate}
\end{proposition}
It actually means that the Prym varieties of spectral curves for $SL(2)$-type Hitchin systems are Jacobians. Classification of the curves possessing this property is given in \cite{Shokur}. The above considered example does not fall under the terms of this classification, at least because the set of singular points and the set of the fixed points of the involution are not the same. The singular point is unique while the branch points coincide with zeros of the polynomial $P_{2g + 1}$, as it follows from the equation \refE{eq3}.
\begin{proof}
Take the equation of the spectral curve in the form $\l^2 - r(x, y) = 0$. The base curve is given by $y^2 = P_{2g + 1}(x)$. Multiplying these equations, we obtain
\begin{equation}\label{E:eq2}
\l^2y^2 + r(x, y)P_{2g + 1}(x) = 0.
\end{equation}
Following the lines of \cite{Gaw}, make a substitution $\s = \l y$:
\begin{equation}\label{E:eq3}
\s^2 + r(x, y)P_{2g + 1}(x) = 0.
\end{equation}
We calculate a genus of the normalization of the obtained curve as of a branch covering of Riemann sphere. The set of its branch points is a union of branch points of the spectral curve and branch points of the base curve, but there is a subtlety: the infinity is a branch point for both curves, so it becomes a simple (in a generic position) singular point. The number of branch points of the spectral curve is equal to $4g - 4$. The base curve has $2g + 2$ branch points (including infinity). Thus, the total number of branch points of the curve \refE{eq3} is equal to $(4g - 4) + (2g + 2) - 2 = 6g -4$ (see the remark below). A genus $\hat{g}$ of the normalization is given by Riemann-Hurwitz formula:
$
2\hat{g} - 2 = 2(h-2) + 6g - 4
$
where $h$ is a genus of the base. In the case in question, $h=0$ which implies $\hat{g} = 3g - 3$.

The holomorphic differentials on \refE{eq3} are given by the relation
\[
\w_j = \frac{(\partial r/\partial H_j)dx}{2\s}, ~j = 1, \ldots, 3g - 3.
\]
Substituting $\s=\l y$ one  obtains
\[
\w_j = \frac{(\partial r/\partial H_j)}{2\l}\frac{dx}{y}, ~j = 1, \ldots, 3g - 3,
\]
that is, the basic Prym differentials on the (standard )spectral curve (according to \refP{prop1}).
\end{proof}
\begin{remark}
We should explain, why the number of branch points of the spectral curve \refE{eq3} is equal to $4g - 4$, including the infinity point. Write $r$ in the form
\[
r(x, y) = a(x) + yb(x),
\]
where $a(x)$ and $b(x)$ are polynomials of the degrees $4g - 4$ and $4g - 5$. Then we represent the equation of the spectral curve in the form
\[
\l^2 + a = -yb,
\]
and set equal the squares of the both parts of it. By that, we get rid of irrationalities but the equation becomes reducible: $(\l^2+a)^2=y^2b^2$, or $(\l^2+a+yb)(\l^2+a-yb)=0$, and one of the irreducible components is the same as the spectral curve itself. At the finite branch points either $\l = 0$ or $\l^2 = -a$. It is easy to realize that the number of the branch points of the first kind is equal to $4g - 4$, and of the second kind to $4g - 5$. As the number of all branch points must be even, we should add infinity as a branch point to one of the irreducible components. On the other hand, irreducible components are symmetric, hence infinity is a branch point for both components.
\end{remark}

%}  % Large
%%%%%%%%%%%%%%%%%%%%%%%%%%%%%%%%%%%%%%%%%%%%%%%%%%%%
\end{document}